\documentclass[reprint,
superscriptaddress,
longbibliography,
showpacs,
preprintnumbers,
nofootinbib,
nobibnotes,
amsmath,
amssymb, 
aps,
prr]{revtex4-1}
\usepackage{blindtext}
\usepackage[T1]{fontenc} 
\usepackage[utf8]{inputenc}
\usepackage[colorlinks=true,linkcolor=blue]{hyperref}%
\usepackage{graphicx}
\usepackage{url}
\usepackage{bm}
\usepackage{booktabs}
\usepackage{lineno}
\usepackage{ulem}
\usepackage{listings}
\usepackage[dvipsnames]{xcolor}
\usepackage[caption=false]{subfig}
\usepackage{tabularx}
\usepackage{lmodern}
\usepackage{gensymb}
\usepackage{fontawesome}
\usepackage{physics}
\usepackage{color}
\usepackage{float}
\usepackage{amsmath}
\usepackage{amsfonts}
\usepackage{amssymb}
\usepackage{relsize}
\usepackage{lettrine}
\usepackage{siunitx}
\usepackage[capitalise]{cleveref}
\usepackage{bbold}
\usepackage{appendix}
\usepackage{etoolbox}
\usepackage{pythonhighlight}
\usepackage{lipsum}

\usepackage{amsthm}
\theoremstyle{definition}
\newtheorem{definition}{Definition}[section]
\newtheorem{proposition}{Proposition}

\input{Zallman.fd}
\begin{document}

\title{Single-shot GHZ characterization with connectivity-aware fanout constructions}

\author{Giancarlo~Gatti}
\email{gatti.gianc@gmail.com}
\affiliation{Basic Sciences Department, Faculty of Engineering, Mondragon Unibertsitatea, Olagorta Kalea 26, 48014 Bilbao, Spain}

\begin{abstract}

We propose a practical recipe to transform any depth-$L$ block of CNOTs that prepares $n$-qubit GHZ states into an $n$-qubit fanout gate (multitarget-CNOT) of depth $2L-1$, without the need for ancilla qubits. Considering known logarithmic-depth circuits to prepare GHZ-states, this allows us to construct an $n$-qubit fanout gate with depth $2\log_2(n)-1$, reproducing previous ancillaless constructions. We employ our recipe to construct $n$-qubit fanout gates under heavy-hex connectivity restrictions, obtaining a depth of $O(n^{1/2})$, again reproducing previous complexity theory constructions. Using this recipe on the \textit{ibm\_fez} architecture yields a $156$-qubit fanout construction with depth $33$. Additionally, we show how to employ these $n$-qubit fanout constructions to measure complete sets of commuting observables from the $n$-body Pauli group with the same depth, allowing for efficient single-shot characterization of any GHZ-like state in a given known basis, e.g. fully characterizing a single copy of a $156$-qubit GHZ state using circuit depth $33$ in $\textit{ibm\_fez}$ (its preparation requires an additional depth of $17$).

 
\end{abstract}

\maketitle

\section{Introduction}
\label{sec:intro}

Fanout gates of $n$ qubits are equivalent to performing CNOT from one control qubit to $n-1$ target qubits. They are powerful multi-qubit gates which can be used to approximate with polynomial error several many-qubit gates, such as \textit{parity}, \textit{mod[q]}, \textit{And}, \textit{Or}\cite{hoyer2005quantum}. Constant depth fanout gates can be constructed using CNOTs, single-qubit gates and mid-circuit measurements with feedforward \cite{baumer2025measurement}. Logarithmic depth fanout gates of $n$-qubits can be constructed with CNOTs and full connectivity \cite{proctor2016ancillas}. Depth $O(n^{1/k})$ fanout gates of $n$ qubits can be constructed with CNOTs and $k$-dimensional grid connectivities \cite{rosenbaum2012optimal}. In this paper, we provide a practical recipe for connectivity-aware fanout constructions, proving its generality and showing it on the \textit{ibm\_fez} architecture for a $156$-qubit fanout gate in depth $33$. In particular, this allows to single-shot measure $156$-body Pauli groups with that same depth.
\section{Definitions}
\label{sec:definitions}

\begin{definition}[$n$-body Pauli observables]
These are observables of the form $\{X,Y,Z\}^{\otimes n}$, where $X = \bigl( \begin{smallmatrix}0 & 1\\ 1 & 0\end{smallmatrix}\bigr)$, $Y = \bigl( \begin{smallmatrix}0 & -i\\ i & 0\end{smallmatrix}\bigr)$ and $Z = \bigl( \begin{smallmatrix}1 & 0\\ 0 & -1\end{smallmatrix}\bigr)$. They have two eigenvalues, $+1$ and $-1$, and thus are degenerate with only two possible measurement outcomes, in contrast to fully-collapsing observables which have $2^n$ possible outcomes. To index and rotate between Pauli observables, we define:
\begin{equation}
\label{eq:pauli_to_pauli}
    \Gamma_0 = \mathbb{1}\,\, , \,\,\Gamma_1=H\,Z\,S\,\, , \,\,\Gamma_2=S\,H,
\end{equation}
\noindent where $H = \tfrac{1}{\sqrt{2}}\bigl( \begin{smallmatrix}1 & 1\\ 1 & -1\end{smallmatrix}\bigr)$ is the Hadamard gate and, $S = \bigl( \begin{smallmatrix}1 & 0\\ 0 & i\end{smallmatrix}\bigr)$ is the Phase gate. The $\Gamma$ rotations can transform Pauli observables in a clockwise or counterclockwise manner:

\begin{equation}
    \begin{array}{rcl}
        \Gamma_0 X \,\Gamma^\dagger_0 = X&\;\;\;\;\Gamma_1 X \,\Gamma^\dagger_1 = Y&\;\;\;\;\Gamma_2 X \,\Gamma^\dagger_2=Z \\
        \Gamma_0 Y \,\Gamma^\dagger_0 = Y&\;\;\;\;\Gamma_1 Y \,\Gamma^\dagger_1 = Z&\;\;\;\;\Gamma_2 Y \,\Gamma^\dagger_2=X \\
        \Gamma_0 Z \,\Gamma^\dagger_0 = Z&\;\;\;\;\Gamma_1 Z \,\Gamma^\dagger_1 = X&\;\;\;\;\Gamma_2 Z \,\Gamma^\dagger_2=Y \text{,}
    \end{array}
\end{equation}
\noindent and we can define a general $\Gamma$ rotation for $n$-body observables:
\begin{equation}
\label{eq:general_gamma}
    {\boldsymbol{\Gamma}}_{\beta} = \bigotimes_{i=1}^n \Gamma_{\beta_i}=\Gamma_{\beta_1}\otimes \Gamma_{\beta_2}\otimes\text{...}\otimes \Gamma_{\beta_n},
\end{equation}
\noindent where $\beta_i=0,1,2$ are single-trit parameters which are composed into $\beta=\overline{\beta_1 \beta_2 \text{...} \beta_n}$, an integer between $0$ and $3^n-1$.

This way, we can identify the $n$-body Pauli Observables with indexes between $0$ and $3^n-1$:
\begin{equation}
\label{eq:general_CSCO}
    \mathcal{O}_\beta={\boldsymbol{\Gamma}}_{\beta}\,Z^{\otimes n}\,{\boldsymbol{\Gamma}}_{\beta}^\dagger,
\end{equation}
\noindent such that $\mathcal{O}_0=Z^{\otimes n}$, $\mathcal{O}_1=Z^{\otimes (n-1)}\otimes X$, $\mathcal{O}_2=Z^{\otimes (n-1)}\otimes Y$, $\mathcal{O}_3=Z^{\otimes (n-2)}\otimes X\otimes Z$, and $\mathcal{O}_{3^n-1}=Y^{\otimes n}$.

\end{definition}

\begin{definition}[$n$-body Pauli contexts]
We define a context as a complete set of commuting observables, and an $n$-body Pauli context as a complete set of commuting $n$-body Pauli observables. We focus on the largest contexts for even $n$, which contain $2^{n-1}+1$ observables and correspond to GHZ-like bases. These contexts are GHZ-class stabilizer groups up to global phase. \cite{Gatti:2021aou,nielsen2010quantum}. This way, the GHZ-like states of these bases form a group under local Clifford rotations, and the $n$-body Pauli contexts themselves also form a group under local Clifford rotations.

In the following, we explicitly show a construction for this class of contexts. To do this, we employ the following two basic contexts, as done in \cite{Gatti:2024thesis,lazar2024new}:
\begin{equation}
        \mathcal{C}^{s}_0 = \{Z^{\otimes n}\}\cup\{X,Y\}^{\otimes n}_{\text{mod}_2(\#X)=s}\text{ ,}
\end{equation}
\noindent where $s=0,1$, and where $\#X$ denotes the number of Pauli $X$ operators in the tensor product. Note that the intersection between $\mathcal{C}^{0}_0$ and $\mathcal{C}^{1}_0$ is a single element: $Z^{\otimes n}$. We define the remaining contexts via the $\Gamma$ rotations:
\begin{equation}
\label{eq:general_CSCO}
    \mathcal{C}^{s}_{\beta}={\boldsymbol{\Gamma}}_{\beta}\,\mathcal{C}^{s}_{0}\,{\boldsymbol{\Gamma}}_{\beta}^\dagger,
\end{equation}
\noindent where the basis transformation given by ${\boldsymbol{\Gamma}}_{\beta}$ acts on each element of $\mathcal{C}^{s}_{0}$. For $n \ge 4$, all $\mathcal{C}^{s}_{\beta}$ are unique~\cite{Gatti:2021aou}.

\end{definition}

\begin{definition}[$n$-qubit GHZ-class states]
In a similar manner, we now define the set of GHZ-like states, which form a group under local Clifford rotations. An $n$-qubit GHZ (Greenberger-Horne-Zeilinger) state is a maximally entangled state of the form
\begin{equation}
\ket{\text{GHZ}_0}=\tfrac{1}{\sqrt{2}}(\ket{0}^{\otimes n}+\ket{1}^{\otimes n})\text{,}
\end{equation}
\noindent which is an eigenstate of the $\mathcal{C}^{0}_{0}$ context of $n$ qubits. To obtain the rest of eigenstates of $\mathcal{C}^{0}_{0}$ and of all $\mathcal{C}_\beta^s$ in general, we apply three types of Clifford local rotations on $\ket{\text{GHZ}_0}$, as done in \cite{Gatti:2024thesis,lazar2024new}: 
\begin{enumerate}
\item We consider a second type of GHZ state, with a phase of $i$, 
\begin{equation}
\ket{\text{GHZ}_1}=\tfrac{1}{\sqrt{2}}(\ket{0}^{\otimes n}+i\ket{1}^{\otimes n})\text{,}
\end{equation}
\noindent which is an eigenstate of the $\mathcal{C}^{1}_{0}$ context of $n$ qubits. We merge both definitions using the phase gate $S$ and a single-bit parameter $s=0,1$ used as exponent to turn it on or off:
\begin{equation}
\ket{\text{GHZ}_s}=(S^s\otimes \mathbb{1}^{\otimes (n-1)})\ket{\text{GHZ}_0}\text{,}
\end{equation}
\noindent This way, $\ket{\text{GHZ}_s}$ is an eigenstate of $\mathcal{C}^{s}_{0}$.

\item We consider all eigenstates of $\mathcal{C}^{s}_{0}$ for $s=0,1$:
\begin{equation}
\label{eq:eigenstates_beta0}
\ket{\phi_{\alpha},\mathcal{C}^{s}_{0}}=Z^{\alpha_1}\otimes\Big(\bigotimes_{i=2}^{n} {X^{\alpha_i}}\Big) \ket{\text{GHZ}_s},
\end{equation}

\noindent where $\alpha_i=0,1$ are single-bit parameters composed into ${\alpha}=\overline{\alpha_1 \, \alpha_2 \, \text{...} \, \alpha_{n}}$, an integer between $0$ and $2^n~-~1$. For each fixed value of $s$, the $\ket{\phi_{\alpha},\mathcal{C}^{s}_{0}}$ states form a basis of $2^n$ eigenstates of $\mathcal{C}^{s}_{0}$.
\item We consider GHZ states in bases other than the computational basis $\{\ket{0},\ket{1}\}$. Specifically, in the Pauli $X$, $Y$ and $Z$ bases: $\{\ket{+},\ket{-}\}$, $\{\ket{L},\ket{R}\}$ and $\{\ket{0},\ket{1}\}$, respectively, where $\ket{+}=\tfrac{1}{\sqrt{2}}(\ket{0}+\ket{1})$, $\ket{-}=\tfrac{1}{\sqrt{2}}(\ket{0}-\ket{1})$, $\ket{L}=\tfrac{1}{\sqrt{2}}(\ket{0}+i \ket{1})$ and $\ket{R}=\tfrac{-1}{\sqrt{2}}(i\ket{0}+ \ket{1})$. These basis changes can be applied on any of the $n$-qubits, so that every $\ket{\phi_{\alpha},\mathcal{C}^{s}_{0}}$ state can be rotated in $3^n$ ways (counting itself) via the $\Gamma_\beta$ rotations defined in Eq.~\eqref{eq:general_gamma}, becoming an eigenstate of context $\ket{\phi_{\alpha},\mathcal{C}^{s}_{\beta}}$. Like before, $\beta=\overline{\beta_1 \beta_2 \text{...} \beta_n}=0,1,\text{...},3^n-1$ with $\beta_i=0,1,2$. Thus, we have 

\begin{equation}
\label{eq:eigenstates_beta0}
\ket{\phi_{\alpha},\mathcal{C}^{s}_{\beta}}=\mathbf{\Gamma}_\beta \ket{\phi_{\alpha},\mathcal{C}^{s}_{0}},
\end{equation}
\end{enumerate}

With all three types of local rotations, our class of GHZ states is given by:
\begin{equation}
\label{eq:eigenstates_beta0}
\ket{\phi_{\alpha},\mathcal{C}^{s}_{\beta}}=\Big(\mathbf{\Gamma}_\beta\Big) \Big((Z^{\alpha_1} S^s)\otimes\bigotimes_{i=2}^{n} {X^{\alpha_i}}\Big)\ket{\text{GHZ}_0}\text{,}
\end{equation}
\noindent which consists of single-qubit Clifford gates applied on a regular GHZ state.

\end{definition}
\section{Generalized fan-out equivalence}
In the context of quantum gates, the fanout gate is defined as the application of CNOT from a single control qubit into multiple target qubits
\begin{equation}
CX_n\equiv\text{CX}({c\rightarrow t_1,t_2,\text{...},t_{n-1}})=\prod_{i=1}^{n-1}\text{CX}(c,t_i)\text{ .}
\end{equation}

A fanout gate can be used to prepare an $n$ qubit GHZ state from a separable state
\begin{equation}
\ket{\psi_0}=\ket{+}_{c}\ket{0}^{\otimes (n-1)}\overset{CX_n}{\longrightarrow}\ket{\psi_{\text{GHZ}}}=\tfrac{1}{\sqrt{2}}(\ket{0}^{\otimes n}+\ket{1}^{\otimes n}),
\end{equation}
but not all gates  
maps $\ket{+}\ket{0}^{\otimes (n-1)}$ into an $n$ qubit GHZ state
As shown in Ref.~\cite{green2001counting}, a transformation that is capable of preparing a GHZ state from $\ket{+}\ket{0}^{\otimes (n-1)}$ is not the same as a fan-out gate, b a fan-out gate can achieve this. Indeed, if a fan-out gate could be prepared with constant depth, then so would a GHZ-state preparation. Ref.~\cite{fang2003quantum} proves a lower bound of depth $O(\log(n))$ for the task of computing an $n$-qubit parity for any constant number of ancilla. This guarantees that the fan-out gate cannot be implemented with lower depth under the same conditions, considering that in principle it could be used to output the parity without additional depth. However, this does not prove that a quantum fan-out gate can be actually implemented with logarithmic depth. In fact, Ref.~\cite{baumer2024measurement} speculates that depth $O(n)$ is required to construct a fan-out gate with no ancillas, even in the case of full-connectivity. Constant depth implementations have been realized \cite{song2024realization,baumer2024measurement}, but they require $O(n)$ ancilla qubits and mid-circuit measurements, which is particularly prohibitive in current hardware. To the best of our knowledge, a unitary implementation of a quantum fan-out gate with no ancillas and depth lower than $O(n)$ still requires to be addressed. 

The fan-out gate satisfies the following equivalence \cite{hoyer2005quantum,baumer2024measurement},
\begin{multline}
\text{CX}(0\rightarrow 1, \text{...},{n-1})=\text{CX}(n-2,n-1)\text{CX}(n-3,n-2)\text{...}\\
\text{CX}(1,2)\text{CX}(0,1)\text{CX}(1,2)\text{...}\text{CX}(n-3,n-2)\text{CX}(n-2,n-1)
\label{fan_out_equivalence}
\end{multline}
\noindent which is also illustrated in Fig.~\ref{fig:fan_out}. In this section we show a generalized version of the fan-out equivalence, allowing us to construct a quantum fan-out gate with twice the depth required in a GHZ-state preparation block. Specifically, our fan-out equivalence generalization transforms a depth-$L$ block of GHZ-constructing CNOTs into a depth-($2L-1$) fan-out gate, that is, as low as depth $2 \log_2(n)-1$ in fully-connected architectures.

\begin{figure}[ht!]
\centering
    \includegraphics[width=0.5\textwidth]{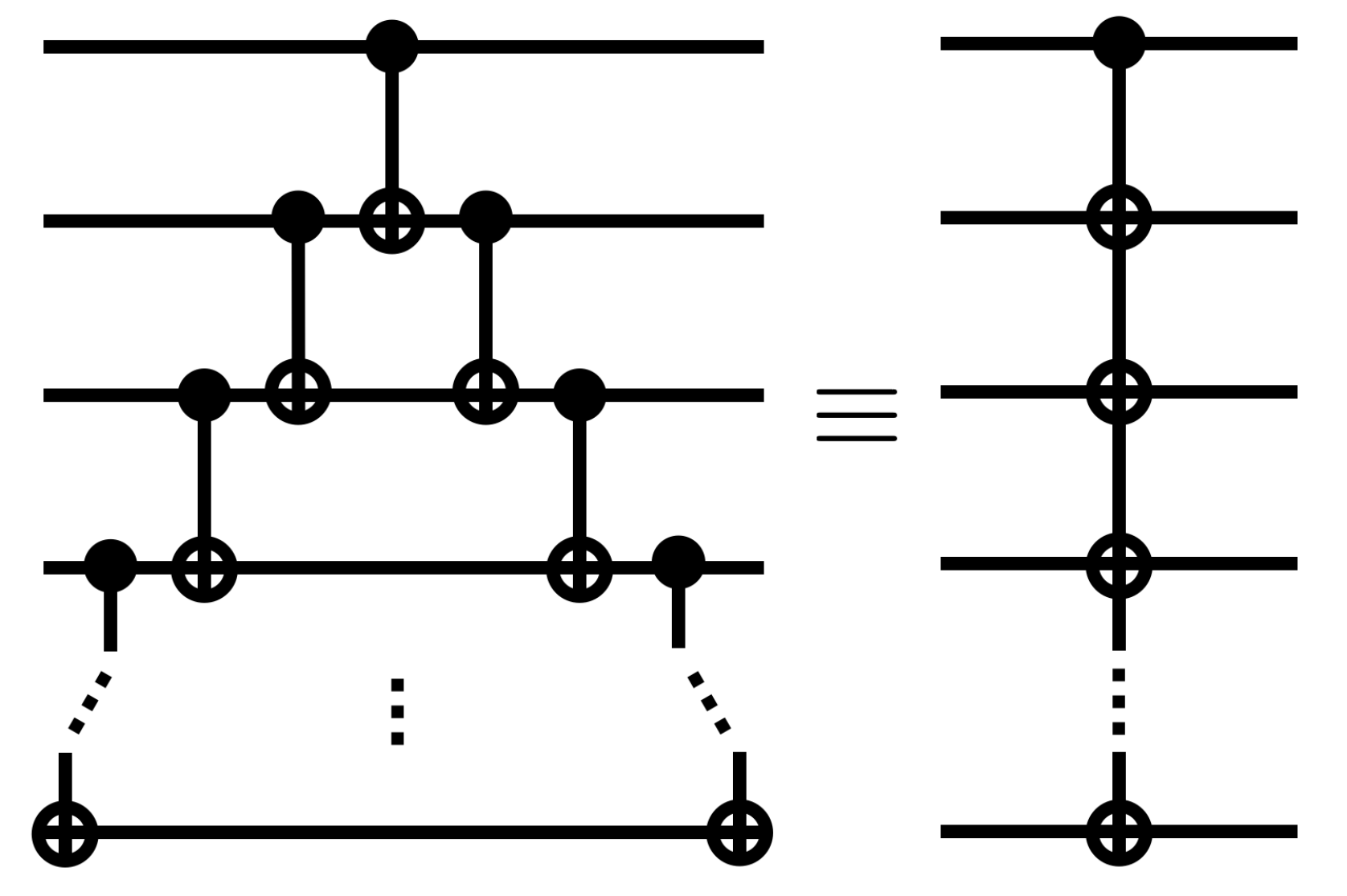}
\caption{\textbf{\textit{Fan-out gate equivalence.}}
A multi-target CNOT, or fan-out gate can be constructed by concatenating two CNOT ladders, thus having an ``open'' and ``closed'' equivalence.
}
\label{fig:fan_out}
\end{figure}

\begin{proposition}
Consider a list $\mathcal{U}_\text{GHZ}$ of $\text{CX}(i,j)$ gates such that their application from left to right maps the $\ket{+}_h\ket{0}_{t_1}\ket{0}_{t_2}\text{...}\ket{0}_{t_{n-1}}$ state into the $n$-qubit GHZ state $\ket{\text{GHZ}_0}=\tfrac{1}{\sqrt{2}}(\ket{0}^{\otimes n}+\ket{1}^{\otimes n})$. Then, the list of CNOTs defined by the concatenation
\begin{equation}
\mathcal{U}_\text{FO}=\text{SEVER}_h\big(\mathcal{U}_\text{GHZ}^\dagger\big)+\mathcal{U}_\text{GHZ}
\label{U_FO}
\end{equation}
\noindent is equivalent under application to a fan-out gate $\text{CX}(h\rightarrow t_1,t_2,\text{...},t_{n-1})$, where $\text{SEVER}_h(\mathcal{U})$ is a function that removes from the list $\mathcal{U}$ all elements with qubit $h$ as control or target, and $\mathcal{U}^\dagger$ reverses the order of list $\mathcal{U}$ and conjugate-transposes each element.
\end{proposition}

\begin{proof}
If $\mathcal{U}_\text{GHZ}$ maps
\begin{equation}
\ket{+}_h\ket{0}_{t_1}\ket{0}_{t_2}\text{...}\ket{0}_{t_{n-1}} \rightarrow \ket{\text{GHZ}_0}=\tfrac{1}{\sqrt{2}}(\ket{0}^{\otimes n}+\ket{1}^{\otimes n})\text{,}
\end{equation}
\noindent then for each qubit $i\neq h$, there exists a subsequence of $\mathcal{U}_\text{GHZ}$ with the form
\begin{multline}
\mathcal{P}^{(i)}=\Big[\text{CX}(h,n^{(i)}_1),\text{CX}(n^{(i)}_1,n^{(i)}_2),\text{CX}(n^{(i)}_2,n^{(i)}_3),\text{ ... },\\
\text{CX}(n^{(i)}_{L^{(i)}-1},n^{(i)}_{L^{(i)}})\Big]\text{,}
\end{multline}
\noindent for some value of $L^{(i)}$ and where $n^{(i)}_{L^{(i)}}=i$. Being a subsequence means that the elements of $\mathcal{P}^{(i)}$ appear in $\mathcal{U}_\text{GHZ}$ with the same application order, but in general with other nonconmuting CNOTs in between, that is, $\mathcal{P}^{(i)}$ is not necessarily a sublist of $\mathcal{U}_\text{GHZ}$. We refer to $\mathcal{P}^{(i)}$ as path $i$ or path to $i$, and define $(\mathcal{P}^{(i)})^c$ as its complement subsequence in $\mathcal{U}_\text{GHZ}$.

We then define $\mathcal{B}^{(i)}$ as the subsequence of $(\mathcal{P}^{(i)})^c$ such that each element of $\mathcal{B}^{(i)}$ is noncommuting with at least one element of $\mathcal{P}^{(i)}$. We call this subsequence the branchings of path $i$. If path $i$ has $b$ branchings located in path steps $\ell_1,\ell_2,\text{...},\ell_{b}$, we have
\begin{equation}
\mathcal{B}^{(i)}=\Big[\text{CX}(n^{(i)}_{\ell_1},x_1),\text{CX}(n^{(i)}_{\ell_2},x_2),\text{...},\text{CX}(n^{(i)}_{\ell_{b}},x_{b})\Big]
\end{equation}
\noindent where $x_j$ are the target qubits of the branchings, such that $x_j$ are not in any control or target of $\mathcal{P}^{(i)}$, but instead define new paths to other qubits.

In any $\mathcal{U}_\text{GHZ}$ list without identical CNOTs in consecutive layers, the subsequence $\mathcal{B}^{(i)}$ corresponds to either CNOT ramifications from path $i$ (with control qubit in path $i$), or a prolongation of path $i$ after reaching qubit $i$ (with control in qubit $i$). This can be visualized in Fig.~\ref{fig:ibm_fez_ghz_layers}, where there is a unique path between qubit $h=89$ and qubit $i=92$, $\mathcal{P}^{(92)}=[\text{CX}(89,90),\text{CX}(90,91),\text{CX}(91,92)]$, and the complementary subsequence of noncommuting CNOTs is $\mathcal{B}^{(92)}=[\text{CX}(91,98),\text{CX}(92,93)]$, which are respectively a ramification and a prolongation of path $92$.

Let us now consider the concatenation $\mathcal{U}_\text{FO}$ defined in Eq.~\eqref{U_FO}. Since there are only commuting CNOTs in $\mathcal{U}_\text{GHZ}$ before $\text{CX}(n^{(i)}_{\ell_1},n^{(i)}_{\ell_1+1})$, and there are only commuting CNOTs in $\text{SEVER}_h\big(\mathcal{U}_\text{GHZ}^\dagger\big)$ after $\text{CX}(n^{(i)}_{\ell_1},n^{(i)}_{\ell_1+1})$, we can reorder $\mathcal{U}_\text{FO}$ such that the following sublist (with nothing in-between) appears between those two elements:
\begin{multline}
\mathcal{F}^{(i)}_1=\Big[\text{CX}(n^{(i)}_{\ell_1-1},n^{(i)}_{\ell_1}),\text{CX}(n^{(i)}_{\ell_1-2},n^{(i)}_{\ell_1-1}),\text{ ...},\\
\text{CX}(n^{(i)}_{1},n^{(i)}_{2}),\text{CX}(h,n^{(i)}_{1}),\text{CX}(n^{(i)}_{1},n^{(i)}_{2}), \text{ ...}\\
\text{CX}(n^{(i)}_{\ell_1-2},n^{(i)}_{\ell_1-1}),\text{CX}(n^{(i)}_{\ell_1-1},n^{(i)}_{\ell_1})\Big]
\end{multline}
\noindent with the rest of $\mathcal{U}_\text{GHZ}$ in the right, and the rest of $\text{SEVER}_h\big(\mathcal{U}_\text{GHZ}^\dagger\big)$ in the left.

We now perform two simplifications with the fan-out equivalence in Eq.~\eqref{fan_out_equivalence} and in Fig.~\ref{fig:fan_out}:
\begin{enumerate}
    \item We identify a V-shape fan-out in $\mathcal{F}^{(i)}_1$, and ``close'' it, obtaining:
    \begin{equation}
    \mathcal{F}'^{(i)}_1=\Big[\text{CX}(h,n^{(i)}_{1}),\text{CX}(h,n^{(i)}_{2}),\text{ ...},\text{CX}(h,n^{(i)}_{\ell_1})\Big]\text{.}
    \end{equation}
    \noindent And we can now see that $\mathcal{F}'^{(i)}_1$ fully commutes with $\mathcal{B}^{(i)}$, with the exception of $\text{CX}(h,n^{(i)}_{\ell_1})$ with $\text{CX}(n^{(i)}_{\ell_1},x_1)$.
    \item If the branching $\text{CX}(n^{(i)}_{\ell_1},x_1)$ comes before the path continuation $\text{CX}(n^{(i)}_{\ell_1},n^{(i)}_{\ell_1+1})$ in $\mathcal{U}_\text{GHZ}$, we can move around commuting elements so that $\mathcal{U}_\text{FO}$ contains the sublist
    \begin{multline}
    \Big[\text{CX}(n^{(i)}_{\ell_1},n^{(i)}_{\ell_1+1}),\text{CX}(n^{(i)}_{\ell_1},x_1),\text{CX}(h,n^{(i)}_{\ell_1}),\\
    \text{CX}(n^{(i)}_{\ell_1},x_1),\text{CX}(n^{(i)}_{\ell_1},n^{(i)}_{\ell_1+1})\Big]\text{,}
    \end{multline}
    \noindent and if the branching comes after the path continuation, we instead move commuting elements so that $\mathcal{U}_\text{FO}$ contains the sublist
    \begin{multline}
    \text{CX}(n^{(i)}_{\ell_1},x_1),\Big[\text{CX}(n^{(i)}_{\ell_1},n^{(i)}_{\ell_1+1}),\text{CX}(h,n^{(i)}_{\ell_1}),\\
    \text{CX}(n^{(i)}_{\ell_1},n^{(i)}_{\ell_1+1}),\text{CX}(n^{(i)}_{\ell_1},x_1)\Big]\text{.}
    \end{multline}
    \noindent By means of the fan-out equivalence, these sublists are application-equivalent (equivalent under application) to
    \begin{equation}
    \mathcal{F}''^{(i)}_1=\Big[\text{CX}(h,n^{(i)}_{\ell_1}),\text{CX}(h,n^{(i)}_{\ell_1+1}),\text{CX}(h,x_1)\Big]\text{,}
    \end{equation}
    \noindent where all CNOTs commute with each other. This can also be done when there are multiple branchings in the same control qubit. This way, the branching in step $\ell_1$ of path $i$ is no longer a problem, and we can remove $\text{CX}(h,x_1)$ from the subsequence of noncommuting CNOTs to $\mathcal{P}^{(i)}$. Within the new application-equivalent version of $\mathcal{U}_\text{FO}$ we should now have the sublist
    \begin{multline}
    \mathcal{F}^{(i)}_2=\Big[CX(n^{(i)}_{\ell_2-1},n^{(i)}_{\ell_2}),CX(n^{(i)}_{\ell_2-2},n^{(i)}_{\ell_2-1}),\text{...}\\
    CX(n^{(i)}_{\ell_1+1},n^{(i)}_{\ell_1+2}),CX(h,n^{(i)}_{\ell_1+1}),CX(n^{(i)}_{\ell_1+1},n^{(i)}_{\ell_1+2}),\text{...}\\
    CX(n^{(i)}_{\ell_2-2},n^{(i)}_{\ell_2-1}),CX(n^{(i)}_{\ell_2-1},n^{(i)}_{\ell_2})\Big]+\\
    \Big[CX(h,n^{(i)}_{1}),CX(h,n^{(i)}_{2}),\text{...},CX(h,n^{(i)}_{\ell_1})\Big]+\\
    [CX(h,x_1)]
    \label{UFO_2}
    \end{multline}
    \noindent where the second and third concatenation summands are all-commuting with the expression and can freely move around. Note that $\mathcal{F}^{(i)}_2$ is a sublist surrounded by the second branch.
\end{enumerate}
We can then repeat steps 1. and 2. to close fan-outs in the first summand of Eq.~\eqref{UFO_2} and incorporate path $i$ up until the third branch. After iteratively repeating this procedure for all $b$ branches, we obtain the $\mathcal{U}_\text{FO}$ sublist
\begin{multline}
\mathcal{F}^{(i)}_b=\Big[CX(h,n^{(i)}_{1}),CX(h,n^{(i)}_{2}),\text{...},CX(h,n^{(i)}_{L})\Big]+\\
[CX(h,x_1),CX(h,x_2),\text{...},CX(h,x_b)]\text{,}
\end{multline}
\noindent that is, a fan-out gate for all qubits of path-$i$ and its branches. Notably, the new $\mathcal{U}_\text{FO}$ still has the property of having a unique path to every qubit, but the length of the path to any qubit in $\mathcal{P}^{(i)}$ and its branchings has been reduced to $1$. This process can then be repeated for all qubits, such that the path length to any qubit (other than $h$) becomes $1$. This way, $\mathcal{U}_\text{FO}$ is shown to be application-equivalent to
\begin{equation}
\mathcal{U}'_\text{FO}=[CX(h,t_1),CX(h,t_2),\text{...},CX(h,t_{n-1})]\text{,}
\end{equation}
\noindent that is, an $n$-qubit fan-out gate.
\end{proof}

In Fig.~\ref{fig:fan_out_examples} we show example circuits of $n$-qubit fan-out gates with $2\log_2(n)-1$ depth. In addition to the analytic proof above, we have numerically tested them to be identical to the fan-out matrix up to $n=16$.

\begin{figure*}[ht!]
\centering
\includegraphics[width=1.0\textwidth]{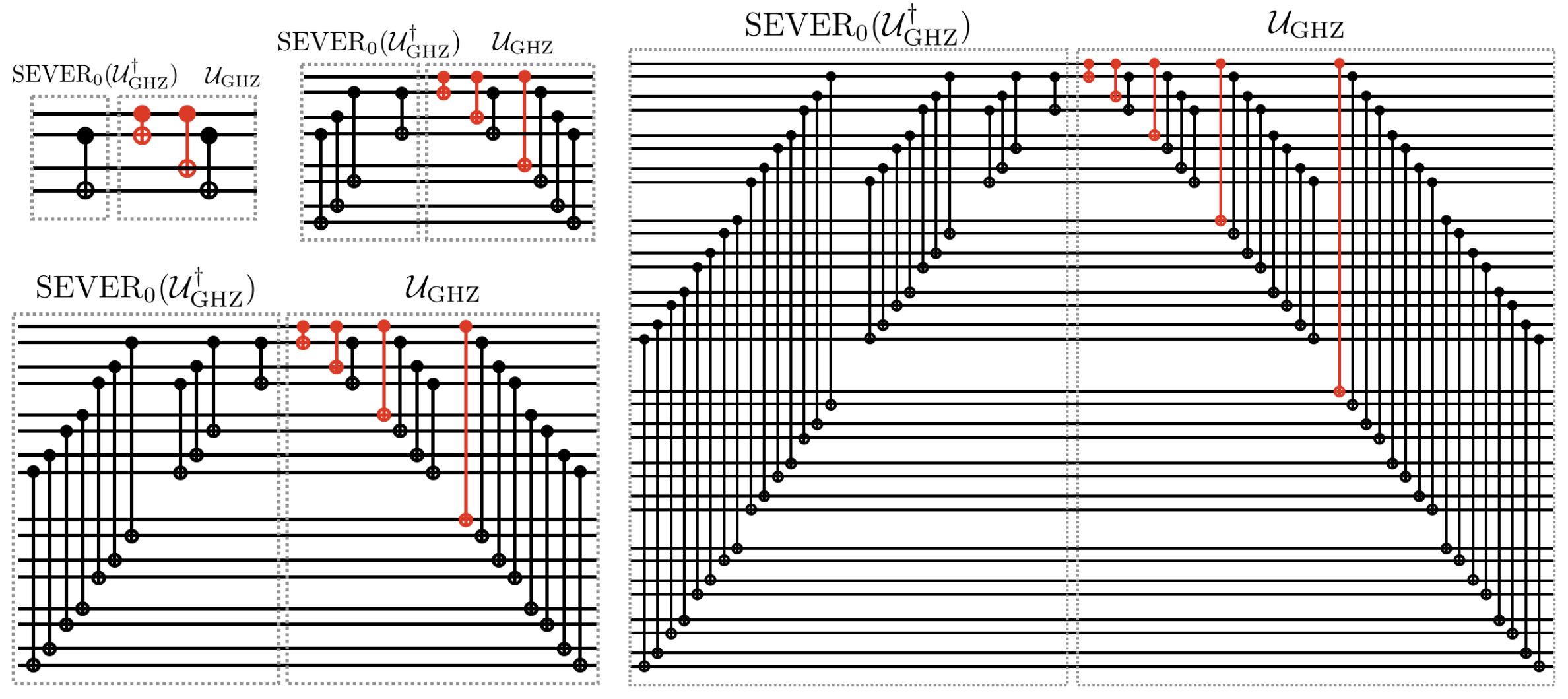}
\caption{\textbf{\textit{Fan-out gate implementation examples.}}
Optimal logarithmic-depth fan-out gate implementations for $n=4,8,16,32$ qubits, with respective depth $D=3,5,7,9$. To produce this construction, we employ a GHZ-generating block of CNOTs $\mathcal{U}_\text{GHZ}$ and its inverse with the CNOTs in red removed.
}
\label{fig:fan_out_examples}
\end{figure*}
\section{GHZ-preparation and context measurement in $156$-qubit heavy-hex topology}
\label{sec:protocol}

To prepare a GHZ-class state, we only need to apply single-qubit gates on a GHZ state, as shown in the previous section. If we want to reduce the noise in this process, we mainly need to optimize the way in which two-qubit gates are applied to prepare the regular GHZ state. In this section, we show how to do this for a $156$-qubit GHZ state in \textit{ibm\_fez}.

A GHZ state of $n$ qubits can be prepared with logarithmic circuit depth \cite{zhang2022quantum}. This can be easily achieved in a $2^k$ qubit quantum architecture with full connectivity, with qubits indexed from $i=0$ to $i=2^k-1$:

\begin{itemize}
\item All qubits are initially in $\ket{0}$. A Hadamard gate is applied on qubit $0$. The system state is: $\ket{\psi_0}=\tfrac{1}{\sqrt{2}}(\ket{0}+\ket{1})\otimes\ket{0}^{\otimes(2^k-1)}$.
\item \textit{Layer $1$:} A CNOT is applied from qubit 0 (control) to qubit 1 (target). We now have a Bell state: $\ket{\psi_1}=\tfrac{1}{\sqrt{2}}(\ket{00}+\ket{11})\otimes\ket{0}^{\otimes(2^k-2)}$.
\item \textit{Layer $2$:} A CNOT is applied from qubit 0 (control) to qubit 2 (target), and a CNOT is applied from qubit 1 (control) to qubit 3 (target). The size of the entangled state has now doubled, and we have a $4$-qubit GHZ state: $\ket{\psi_2}=\tfrac{1}{\sqrt{2}}(\ket{0000}+\ket{1111})\otimes\ket{0}^{\otimes(2^k-4)}$.
\item \textit{Layer $i$:} CNOTs are applied from the $2^{i-1}$ qubits in the GHZ state (controls) to $2^{i-1}$ still-unentangled qubits. The size of the GHZ state doubles again, from $2^{i-1}$ qubits to $2^{i}$ qubits: $\ket{\psi_i}=\tfrac{1}{\sqrt{2}}(\ket{0}^{\otimes 2^i}+\ket{1}^{\otimes 2^i})\otimes\ket{0}^{\otimes(2^k-2^i)}$.
\end{itemize}

The two-qubit layers of this circuit can be grouped in the following block of CNOTs:
\begin{equation}
U_\text{GHZ}=\prod_{i=1}^{\log_2{n}}\prod_{j=1}^{i}\text{CNOT}(j-1,2^{i-1}+j-1)\text{,}
\label{GHZvirusblock}
\end{equation}
\noindent where the products are expressed from right to left, so that the $i$ index corresponds to the layer number. Note that $U_\text{GHZ}$ is different from a multi-target CNOT $C\!X^{(n)}$, but both can construct a GHZ state from the input state $\ket{+}\ket{0}^{\otimes (n-1)}$.

As we can see, if we have access to full connectivity, $k$ layers are sufficient to prepare a GHZ state of $2^k$ qubits. In reality, however, we need to adapt our approach to grow the GHZ state size as fast as possible with the two-qubit connectivity constraints of the actual device. In Fig.~\ref{fig:ibm_fez_ghz_layers} we show the \textit{heavy-hex} connectivity map of \textit{ibm\_fez} and the $17$-layers of two qubit gates used to generate a $156$-qubit GHZ state. Note that we begin by applying the Hadamard gate on qubit $89$, which has $3$ neighbours, but could as easily start with qubit $78$ ($2$ neighbours) instead, since layer $1$ makes a Bell state with them anyway. The size of GHZ-states that can be obtained for various depths is detailed in Table~\ref{table:depth_vs_ghzsize}.

The two-qubit layers of this circuit can be grouped in a block of CNOTs of the following form:
\begin{multline}
U^{\text{hex}}_\text{GHZ}=\Big(\text{CX}(154,155)\;\text{...}\;\text{CX}(4,5)\;\text{CX}(1,0)\Big)\text{...}\\
\Big(\text{CX}(78,69)\;\text{CX}(89,88)\Big)\;\Big(\text{CX}(89,78)\Big)\text{,}
\label{GHZvirusblock}
\end{multline}
\noindent where each parenthesis contains a layer. The rightmost parenthesis is layer $1$, and the leftmost parenthesis is layer $17$. For the sake of space, we have omitted most of the CNOTs and layers in the expression as there are too many, but they can be retrieved from Fig.~\ref{fig:ibm_fez_ghz_layers}.

\begin{table}[h]
    \centering
    \begin{tabular}{|c|c||c|c|}
        \hline
        Depth & GHZ size (qubits) & Depth & GHZ size (qubits) \\
        \hline\hline
        1       & 2 \text{(Bell)}   & 10       & 70 \\
        2       & 4       & 11       & 84       \\
        3       & 7       & 12       & 99       \\
        4       & 11       & 13       & 112       \\
        5       & 16       & 14       & 124       \\
        6       & 23       & 15       & 136       \\
        7       & 32       & 16       & 148       \\
        8       & 44       & 17       & 156       \\
        9       & 56       &        &        \\
        \hline
    \end{tabular}
    \caption{Relation between the number of two-qubit-gate layers in the quantum circuit and the size of GHZ states that can be obtained (Bell state in the case of $2$ qubits).}
    \label{table:depth_vs_ghzsize}
\end{table}

\begin{figure*}[ht!]
\centering
    \includegraphics[width=0.75\textwidth]{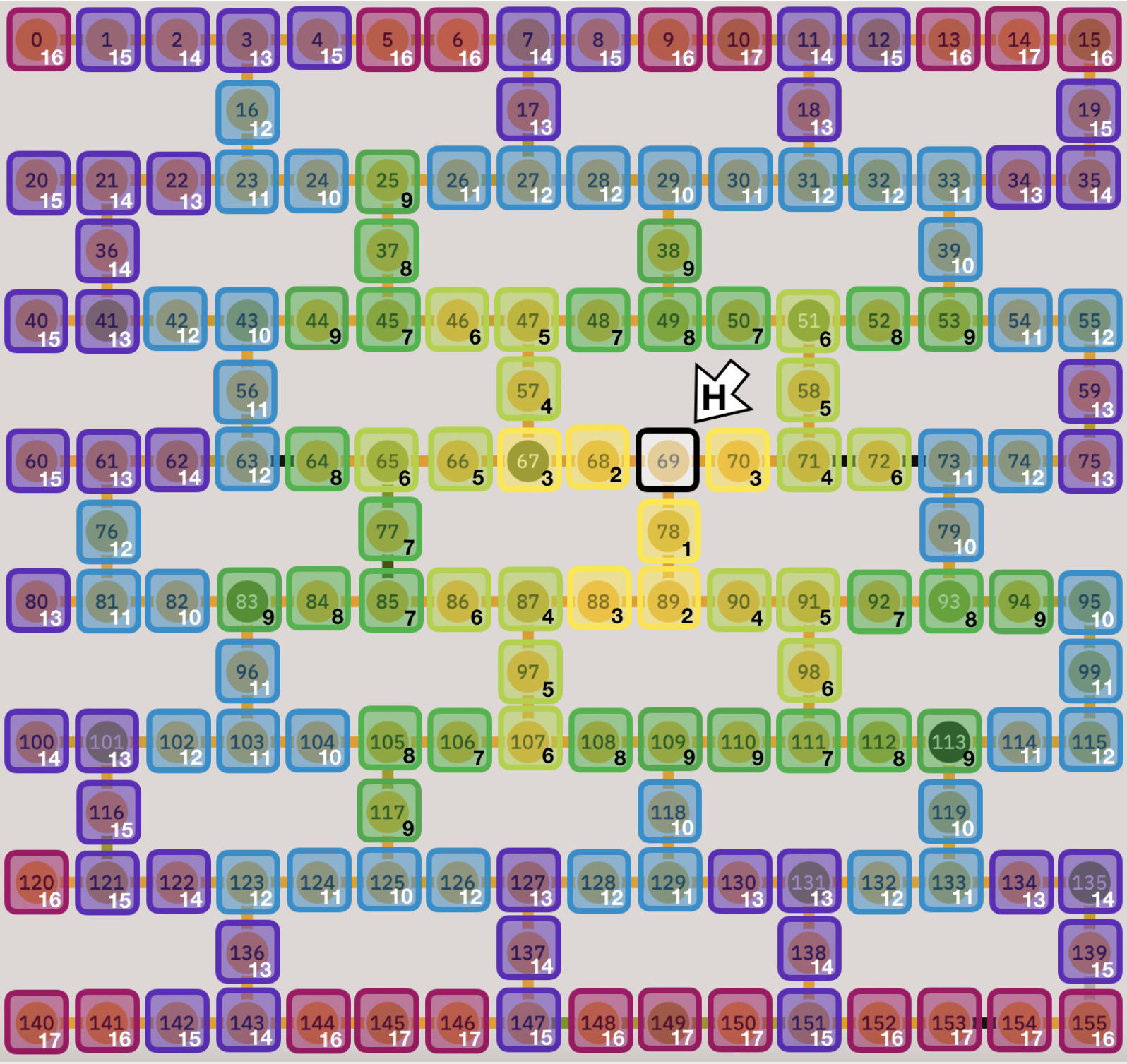}
\caption{\textbf{\textit{Connectivity and layer map for GHZ-state preparation in \textit{ibm\_fez}.}}
To prepare a $156$-qubit GHZ state, a Hadamard gate is applied on qubit number $69$ followed by $155$ CNOTs, which can be grouped into $17$ layers. In this map, the qubits are indexed between  $0$ and $155$, and all qubits excepting qubit $69$ are also accompanied by a number between $1$ and $17$, which indicates the layer in which they are targeted by a CNOT, with one of their entangled neighbours being the control. Additionally, we color-code groups of consecutive layers for easier visualization. 
}
\label{fig:ibm_fez_ghz_layers}
\end{figure*}
\label{sec:conclusion}
\section{Conclusions and Outlook}
We have demonstrated a process to transform any list of GHZ-constructing CNOTs into a fanout gate of twice the depth (minus $1$). Here, $\mathcal{U}_\text{GHZ}$ is understood to map $\ket{+}\ker{0\text{...}0}$ into a GHZ state, which is something that a fanout gate can do, but which does not automatically make $\mathcal{U}_\text{GHZ}$ a fanout gate. Our process is in agreement with previous results proving that fanout could be constructed with logarithmic depth under full-connectivity and with $O(n^{1/k})$ depth for $k$-dimensional connectivity, as it is known that GHZ states can be prepared with these depths under these connectivity restrictions. We illustrate this by showing how to prepare a $156$-qubit GHZ state in \textit{ibm\_fez} (heavy-hex connectivity) using depth $17$. This allows us to construct a $156$-qubit fanout gate with depth $33$, which can be used to fully measure any context in the $n$-body Pauli group with that same depth, allowing single-shot characterization of GHZ-like states.



\section*{Acknowledgments}
We acknowledge support from the Basque Goverment ELKARTEK KUBIBIT project with reference KK-2025/00079.

\bibliography{qencoding.bib}

\pagebreak
\clearpage
\onecolumngrid
\appendix

\ifx \standalonesupplemental\undefined
\setcounter{page}{1}
\setcounter{figure}{0}
\setcounter{table}{0}
\setcounter{equation}{0}
\fi

\renewcommand{\thepage}{Supplemental Methods and Tables -- S\arabic{page}}
\renewcommand{\figurename}{SUPPL. FIG.}
\renewcommand{\tablename}{SUPPL. TABLE}

\AtBeginEnvironment{appendices}{\crefalias{section}{appendix}}



\end{document}